\documentclass[journal,12pt,onecolumn,draftclsnofoot]{IEEEtran}

\IEEEoverridecommandlockouts
\usepackage{cite}
\usepackage{amsmath,amssymb,amsfonts}
\usepackage{algorithmic}
\usepackage{graphicx}
\usepackage{subcaption}
\usepackage{url}
\usepackage{bbm}
\usepackage{bm}
\usepackage{siunitx}
\usepackage{mathtools, nccmath}
\graphicspath{ {Figures/EPS/} }
\usepackage{textcomp}
\usepackage{xcolor}
\usepackage{mathtools}
\usepackage{blindtext}
\usepackage[utf8]{inputenc}
\usepackage{float}
\usepackage{enumitem}
\usepackage{hyperref}
\usepackage{graphicx}
\usepackage{multirow}
\usepackage{tabularx}
\usepackage{amsthm}

\usepackage{acronym}
\acrodef{vlc}[VLC]{visible light communication}
\acrodef{mumiso}[MU-MISO]{multi-user multiple-input single-output}
\acrodef{csi}[CSI]{channel state information}
\acrodef{rf}[RF]{radio frequency}
\acrodef{socp}[SOCP]{second order cone programming}
\acrodef{led}[LED]{light emitting diode}
\acrodef{ul}[UL]{uplink}
\acrodef{dl}[DL]{downlink}
\acrodef{tdd}[TDD]{time division duplexing}
\acrodef{fdd}[FDD]{frequency division duplexing}
\acrodef{mimo}[MIMO]{multiple-input multiple-outpout}
\acrodef{mse}[MSE]{mean square error}
\acrodef{pd}[PD]{photo diode}
\acrodef{dc}[DC]{direct current}
\acrodef{snir}[SNIR]{signal to noise plus interference ratio}
\acrodef{snr}[SNR]{signal to noise ratio}
\acrodef{ser}[SER]{symbol error rate}
\acrodef{pdf}[PDF]{probability density function}
\acrodef{los}[LOS]{line of sight}
\acrodef{nlos}[NLOS]{no line of sight}
\acrodef{ppp}[PPP]{Poisson point process}
\acrodef{bs}[BS]{base station}
\acrodef{mmw}[mmWave]{millimeter waves}
\acrodef{rv}[r.v.]{random variable}
\acrodef{zf}[ZF]{zero forcing}

\newtheorem{theorem}{Theorem}

\newtheorem{lemma}[theorem]{Lemma}

\begin{document}
\title{Robust Precoding for Multi-User Visible Light Communications with Quantized Channel Information}

\author{Olga Muñoz, Antonio Pascual-Iserte, Guillermo San Arranz \\Department of Signal Theory and Communications, Universitat Politècnica de Catalunya, 08034 Barcelona, Spain \\ olga.munoz@upc.edu (O.M.); antonio.pascual@upc.edu (A.P.-I.); guille.san.arranz@gmail.com (G.S.A.)
\thanks{These authors contributed equally to this work.\\This work has been funded by the Agencia Estatal de Investigación, Ministerio de Ciencia e Innovación (Gobierno de España), MCIN/AEI/10.13039/501100011033, through the project ROUTE56 PID2019-104945GB-I00.\\

DOI: 10.3390/s22239238}}

\markboth{Accepted paper at MDPI Sensors (Sensors 2022, 22, 9238)}
{}

\maketitle

\begin{abstract}
In this paper, we address the design of \ac{mumiso} precoders for indoor \ac{vlc} systems. The goal is to minimize the transmitted optical power per \ac{led} under imperfect \ac{csi} at the transmitter side. Robust precoders for imperfect \ac{csi} available in the literature include noisy and outdated channel estimation cases. However, to the best of our knowledge, no work has considered adding robustness against channel quantization. In this paper, we fill this gap by addressing the case of imperfect \ac{csi} due to the quantization of \ac{vlc} channels. We model the quantization errors in the \ac{csi} through polyhedric uncertainty regions. For polyhedric uncertainty regions and positive real channels, as is the case of \ac{vlc} channels, we show that the robust precoder against channel quantization errors that minimizes the transmitted optical power while guaranteeing a target \ac{snir} per user is the solution of a \ac{socp} problem. Finally, we evaluate its performance under different quantization levels through numerical simulations.
\end{abstract}









\section{Introduction}
\subsection{Background and Motivation}
In the past years, researchers have proposed different multi-user precoding techniques for the visible light communication (VLC) channel, assuming the availability of channel state information (CSI) at the transmitter side. In \ac{rf} systems, there are methods to acquire the \ac{csi}, such as those based on the exploitation of channel reciprocity between \ac{ul} and \ac{dl} in \ac{tdd} or the use of a feedback channel in \ac{fdd}. Unfortunately, establishing an \ac{ul} feedback channel in \ac{vlc} is not straightforward. One possibility, which does not relay on channel reciprocity, is to send the \ac{ul} signals through an \ac{rf} channel. In this case, once the receivers have estimated the \ac{dl} channel, the estimates can be fed back to the transmitter through the \ac{rf} \ac{ul} link. Consequently, the \ac{csi} available at the transmitter side will be far from being accurate, as it can be noisy, outdated, and will contain quantization errors.

In this paper, we address the design of a multi-user multiple-input single-output (MU-MISO) precoder robust against channel quantization errors, which will be the dominant source of error in the \ac{csi} if the \ac{snr} during the channel estimation is high or if the number of quantization bits is low. We model the quantization errors in the \ac{csi} through polyhedric uncertainty regions that contain the actual channel. Interestingly, for polyhedric uncertainty regions and the fact that, for \ac{vlc} communications, the transmitted signals and the channels are positive real magnitudes allows us obtaining a non-approximated solution for the robust precoder. This design has not been developed before for \ac{rf} or \ac{vlc} channels. In particular, for \ac{vlc} channels, the only available solutions are those considering estimation noise or outdated \ac{csi}, which do not fit the model corresponding to quantization errors, as explained in what follows.

\subsection{Related Work}
In Ref. \cite{RoserGlobecom}, the authors consider a \ac{mumiso} \ac{vlc} system in a scenario where the passengers in a train wagon receive data through the light emitting diodes (LEDs) placed at the ceiling. Similarly, in Refs. \cite{antiguo, ZHAO2017341,ZFperfectCSI,9463422,rev1} the authors address the design of the \ac{dl} in VLC systems under the perfect CSI assumption. In this paper, as mentioned before, we consider a \ac{mumiso} scenario similar to those considered in the previous references (and, in particular, in Ref. \cite{RoserGlobecom}) but assuming that the CSI is not perfect. The scenario is summarized in Section \ref{sec:Scenario} for completeness.

Imperfect \ac{csi} has been considered in Refs. \cite{RobustMMSE, noisyCSI, rev3, rev2}. The authors of Refs. \cite{RobustMMSE,noisyCSI} model the channel error as zero-mean Gaussian distributed with a certain covariance matrix to account for noisy channel estimation.
The authors of Ref. \cite{rev3} impose an upper bound on the Frobenius norm of the channel error to account for outdated \ac{csi}. Finally, the authors of Ref. \cite{rev2} consider both the stochastic and bounded error model to account for both channel estimation errors and outdated \ac{csi}. None of these papers \cite{RobustMMSE, noisyCSI, rev3, rev2}, however, have considered channel estimation errors due to the quantization of the channel response. In this paper, we fill this gap by considering imperfect CSI due to the quantization of VLC channels. Note that, in the absence of channel reciprocity, for example if the \ac{ul} cannot be optical, the receiver needs to report the channel response through a digital feedback channel. This makes quantization a primary CSI error source, particularly when the channel estimation \ac{snr} is good or when the number of quantization bits is low.

Finally, regarding the precoder designs, some authors design \ac{zf} linear precoders while maximizing the weighted sum rate or other figures of merit, e.g., \cite{RoserGlobecom,rev1}. Other works do not relay on \ac{zf} and, instead, focus on minimizing the global \ac{mse}, e.g., \cite{RobustMMSE}, or the average \ac{mse} and the worst-case \ac{mse}, e.g., \cite{rev2}, for the noisy \ac{csi} and the outdated \ac{csi}, respectively. Similarly, in Refs. \cite{noisyCSI,rev3}, the authors address the maximization of the minimum signal to noise plus interference ratio (SNIR), also for the cases of noisy and outdated \ac{csi}, respectively. In our work, instead of \ac{zf} precoders as in Refs. \cite{RoserGlobecom,rev1}, we minimize the optical power per \ac{led} subject to a target \ac{snir} per user without forcing null interference, which makes the design more versatile, and account explicitly for quantization errors. On the other hand, instead of considering a global figure of merit that cannot guarantee the quality of individual users, such as the global or the average \ac{mse}, we deal with individual user \ac{snir} constraints.

\subsection{Main Contributions}
Our main contributions and novelties in this paper are the following:
\begin{itemize}
\item  The design of a robust \ac{mumiso} precoder for \ac{vlc} systems under imperfect \ac{csi} at the transmitter due to channel quantization. We show that for \ac{vlc} with \ac{csi} imperfections due to quantization, the robust precoder (optimum in terms of transmitted optical power per \ac{led}) is the solution of a second order cone programming (SOCP) problem \cite{Boyd04};
\item  Quantization guidelines and discussion of their impact on the
beamforming design;
\item The evaluation of the performance of the robust precoder under different levels of quantization. 
\end{itemize}

\section{System Model}
\label{sec:Scenario}
We consider a \ac{vlc} system, where a transmitter equipped with $L$ \ac{led}s communicates simultaneously with $K$ users (with indexes $k=1,\ldots, K$) deployed in the same room. 

The transmitter modulates the intensity of the light emitted by each \ac{led} according to the information symbols. Let $\mathbf{s}=(s_1,s_2,...,s_K)^T \in \mathbb{R}^{K \times 1}$ be the vector containing the information-bearing real symbols, one per user, at time $t$, although the time index $t$ is omitted for simplification. We model, without loss of generality, each symbol as a \ac{rv} taken from an alphabet with zero mean $\mathbb{E}\{s_k\}=0$ and unit variance $\mathbb{E}\{s_k^2\}=1$. We assume that the dynamic range of each symbol is limited so that $-A_k\leq s_k \leq A_k$.

Using a set of real beamvectors $\mathbf{w}_k \in  \mathbb{R}^{L \times 1}$, $k=1,...,K$, the transmitter combines the set of transmitted symbols at time $t$, $\{s_k\}_{k=1}^K$. Therefore, the vector  $\mathbf{x} \in  \mathbb{R}^{L \times 1}$  containing the $K$ optical transmitted symbols at time $t$ (omitted again for simplification) can be written as
\begin{equation}
\mathbf{x}= \sum_{i=1}^{K} \mathbf{w}_i s_i+\beta\mathbf{1}\in \mathbb{R}_+^{L \times 1},
\end{equation}
where $\beta$ is the \ac{dc} component equal to the average optical power emitted per LED. Each component of the vector $\mathbf{x}$ , at any given instant $t$, corresponds to the optical transmit power per individual LED, which is a positive physical magnitude and is usually limited due the technological constraints and eye safety reasons. Therefore, the  set of precoders $\{\mathbf{w}_k\}_{k=1}^K$ should be designed to guarantee that, for any combination of symbols, every component of $\mathbf{x}$ is positive and below the maximum allowed optical power, $P_{\max}$.  On the other hand,  $\beta$ determines the average light intensity. Because of this, in practice, $\beta$ is a preset value, equal for all the LEDs, and we will assume so in this work.

At the receiver side, each user is equipped with a single \ac{pd} that provides an electrical signal according to the incident light and the responsivity $\rho$ of the \ac{pd}  (i.e., the ratio of the output photocurrent to the input optical power measured in [A/W]). The incident light depends on the light emitted by the set of \ac{led}s and the channel gain between  the $l$-th  \ac{led} and the $k$-th \ac{pd} denoted by $h_{k,l}\in\mathbb{R}_+$ (a magnitude which is, again, real and positive). The channel gain, $h_{k,l}$, depends on the distance between the $l$-th \ac{led} and the $k$-th \ac{pd}, the incident angle, and the irradiation angle, as well as other parameters (see \cite{RoserGlobecom} for a full description). The column vector containing all the channel gains from all the \ac{led}s to the $k$-th user's \ac{pd}  is given by $\mathbf{h}_k=(h_{k,1},h_{k,2},...,h_{k,L})^T \in \mathbb{R}_+^{L \times 1}$.

Accordingly, the electrical signal at the output of the $k$-th \ac{pd}, $y_k \in \mathbb{R}$, is equal to 
\begin{equation}
y_k=\rho\mathbf{h}_k^T\mathbf{x}+n_k=\rho\mathbf{h}_k^T\sum_{i=1}^K \mathbf{w}_is_i+\rho\beta\mathbf{h}_k^T\mathbf{1}+n_k,\label{eq:rx_signal}
\end{equation}
where $n_k$ is the real-valued additive white Gaussian noise with zero mean and variance $\sigma_k^2$ modeling the thermal and shot noise \cite{haas}. The term $\beta\mathbf{h}_k^T\mathbf{1}$ is a \ac{dc} component that can be easily estimated and removed by the receiver as it does not provide any information about the transmitted symbols. After removing the \ac{dc} component, the electrical \ac{snir} at the $k$-th \ac{pd} can be written as (note that the receiver converts the optical power to current according to the responsivity factor $\rho$, which implies that the receiver electrical \ac{snr} in \ac{vlc} is proportional to the square of the received optical average power, while in RF it is directly proportional to the received average power, as explained in Ref. \cite{RoserGlobecom}):
\begin{equation}
\mbox{SNIR}_k=\frac{\rho^2|\mathbf{h}_k^T\mathbf{w}_k|^2}{\sigma_k^2+\sum_{i\neq k}\rho^2|\mathbf{h}_k^T\mathbf{w}_i|^2}. \label{eq:SNIR}
\end{equation}

The \ac{ser} depends on the \ac{snir}. This means that, if a minimum detection quality in terms of a maximum \ac{ser} is required, this can be translated equivalently into a minimum \ac{snir} requirement per user denoted by $\gamma_k$. Note that the relation between the maximum \ac{ser} and $\gamma_k$ depends on the size of the adopted modulation strategy (e.g., $M$-PAM), that is, the transmission rate \cite{haas}. If different users use different modulations, the corresponding minimum required \ac{snir}s will also differ.

\section{Problem Formulation}
\label{sec:Problem}

In this section, we address the design of a robust \ac{mumiso} precoder under imperfect \ac{csi} at the transmitter side for the \ac{vlc} system described Section \ref{sec:Scenario}. We assume that the receivers report the \ac{dl} channels to the transmitter through a feedback channel. Accordingly, the informed channels may differ from the actual channels due to estimation and reporting errors:
\begin{itemize}
\item  If the source of \ac{csi} imperfection is the quantization, the actual channel will be within an uncertainty region whose shape and size depend on the quantization type (scalar or vector quantization) and the number of quantization bits;

\item If the source of imperfection is an outdated \ac{csi}, the effect can be modeled by imposing an upper bound on the Frobenius norm of the channel error, as in Refs. \cite{rev3,rev2}. In this case, the actual channel will be within a spherical uncertainty region;

\item If the source of \ac{csi} imperfection is the channel estimation noise (assumed as Gaussian in Refs. \cite{RobustMMSE,noisyCSI}), the norm of the error in the \ac{csi} will be upper-bounded by a threshold with a certain probability $p_{in}<1$ and, therefore, the corresponding uncertainty region will be spherical. The value of the threshold depends on the estimation \ac{snr} and the value of $p_{in}$.
\end{itemize}

If the three previous sources for \ac{csi} imperfection are simultaneously present, the global uncertainty region is a convolutive combination of the different regions \cite{pascual_2006}.
In this paper, we do not consider the channel estimation noise and the outdated channel information and assume that the quantization effect is dominant, which is a valid assumption in the case of having good channel estimation \ac{snr} or a low number of quantization bits and when the channel does not vary quickly. For simplicity, we also assume that each receiver $k$ quantizes the acquired channel response $\mathbf{h}_k$ independently. 

For a channel quantizer based on \emph{vector quantization} and characterized by a set of $N$ reconstruction points $\{\mathbf{c}_i \in \mathbb{R}^{L\times 1},i=1,\ldots,N\}$, the uncertainty region associated to each reconstruction point (or centroid) $\mathbf{c}_i$ is defined as:
\begin{equation}
\mathcal{R}_i=\{\mathbf{h} \in \mathbb{R}_+^{L\times 1} \quad | \quad ||\mathbf{h}-\mathbf{c}_i|| \leq ||\mathbf{h}-\mathbf{c}_j|| \quad \forall j\neq i\}.
\end{equation}

The centroids $\{\mathbf{c}_i\}_{i=1}^N$ correspond to the $N$ codewords composing the quantizer codebook. Therefore, each receiver requires sending ceil$({\log_2(N)})$ bits for \ac{csi} feedback. Note that $\mathcal{R}_i$ is the Voronoi partition around the centroid  $\mathbf{c}_i$ with respect to (w.r.t.) the set $\{\mathbf{c}_i\}_{i=1}^N$. As the space is a finite-dimensional Euclidean space, these Voronoi cells are convex polygons completely defined by their vertices \cite{Boyd04}, that is, polyhedric uncertainty regions.

The previous notation encompasses, as a particular case, \emph{scalar quantization}, that is, an independent quantization of each component of the channel vector. In this particular case, the uncertainty regions are rectangular and the number of vertices of each region is $2^L$.

In what follows, we refer to the uncertainty region of the channel informed by the $k$-th user as  $\mathcal{H}_k$. For the $k$-th user, $\mathcal{H}_k$ will be one of the Voronoi regions within the set $\{\mathcal{R}_1,...,\mathcal{R}_N\}$. Each region is a convex polygon that can be expressed as the convex hull \cite{Boyd04} of its $J_k$ vertices $\mathbf{h}_k^{(j)},j=1,\ldots,J_k$, that is: $\mathcal{H}_k=\mbox{coh}(\{\mathbf{h}_k^{(j)},j=1,\ldots,J_k\})$.

Having defined the uncertainty regions, in the following we formulate the design of a precoder fulfilling the \ac{snir} constraint per user for any possible channel in the uncertainty region: 

\begin{equation}
\begin{aligned}
 \underset{\{\mathbf{w}_k\},v}{\text{minimize  }}
 & v \\
 \text{subject to  }
 & C1: \frac{\rho^2 |\mathbf{h}_k^T\mathbf{w}_k|^2}{\sigma_k^2+\sum_{i\neq k}\rho^2 |\mathbf{h}_k^T\mathbf{w}_i|^2}\geq \gamma_k, \; \quad \forall \mathbf{h}_k \in \mathcal{H}_k, \quad k = 1, \ldots, K;\\
 & C2: \sum_k{ A_k|\mathbf{e}_l^T\mathbf{w}_k|} \leq v, \; \quad l = 1, \ldots, L; \\
 & C3: 0\leq v \leq \min(\beta, P_{\max}-\beta).
\end{aligned}
\label{eq:imperfect_CSI_0}
\end{equation}

The C1 constraints impose the $K$ minimum \ac{snir} requirements (one per user). In C2, $\mathbf{e}_l= (0, ..., 0, 1, 0, ..., 0)^T$ is a vector whose $l$-th element is equal to 1, and the rest of elements are zero.  The C2 constraints imply that the optical power transmitted per LED is within the interval $[\beta-v,\beta+v]$. The C3 constraint ensures that $[\beta-v,\beta+v]$ is within the interval $[0, P_{\max}]$, and, therefore, the optical power transmitted per LED is positive and not greater than $P_{\max}$.

With some basic manipulations, we can re-write the C1 constraints as follows:
\begin{equation}
C1:\sqrt{\sigma_k^2+\sum_{i\neq k}\rho^2 |\mathbf{h}_k^T\mathbf{w}_i|^2}  \leq \frac{\rho}{\sqrt{\gamma_k}}|\mathbf{h}_k^T\mathbf{w}_k|, \quad \forall \mathbf{h}_k \in \mathcal{H}_k, \quad k = 1, \ldots, K.
\label{eq:alternative_C1}
\end{equation}

Both alternative expressions of constraints C1 (either in \eqref{eq:imperfect_CSI_0} or in \eqref{eq:alternative_C1}) represent $K$ sets of infinite non-convex constraints because the constraint expressed in C1 for each $k$ has to be fulfilled for the infinite set of channels represented by $\mathcal{H}_k$.

\section{Solution for Polyhedric Uncertainty Regions}
\label{sec:Solution}
In this section, we show that, for polyhedric uncertainty regions $\mathcal{H}_k$, problem \eqref{eq:imperfect_CSI_0} can be re-formulated as a convex problem with a finite number of second-order cone constraints. Therefore, the problem becomes a \ac{socp} problem, for which several efficient interior point methods are available \cite{Boyd04}.

As a first step, we show in Lemma \ref{lemma:1} that we can force, without any loss of generality, $\mathbf{w}_k^T\mathbf{h}_k$  to be positive for any channel within the region $\mathcal{H}_k$. 

\begin{lemma}\label{lemma:1}
Any feasible precoder $\mathbf{w}_k$ will be such that $\mathbf{w}_k^T\mathbf{h}_k$ will be always positive or always negative for any $\mathbf{h}_k \in \mathcal{H}_k$.
\end{lemma}

\begin{proof}
We will prove it by contradiction. Let us assume that there is a feasible $\mathbf{w}_k$, that is, a vector $\mathbf{w}_k$ that satisfies constraints C1, C2, and C3, and two channels in $\mathcal{H}_k$, namely $\mathbf{h}_k^{(1)},\mathbf{h}_k^{(2)}\in\mathcal{H}_k$, such that $\mathbf{w}_k^T\mathbf{h}_k^{(1)} >0$ and $\mathbf{w}_k^T\mathbf{h}_k^{(2)} <0$. Therefore, we can find  $\alpha$ such that $\mathbf{w}_k^T\mathbf{\bar{h}}_k=0$ with $\mathbf{\bar{h}}_k=\alpha \mathbf{h}_k^{(1)}+(1-\alpha)\mathbf{h}_k^{(2)}$ and $0<\alpha<1$. In fact, we have:
\begin{equation}
\alpha=\frac{-\mathbf{w}_k^T\mathbf{h}_k^{(2)}}{\mathbf{w}_k^T\mathbf{h}_k^{(1)}-\mathbf{w}_k^T\mathbf{h}_k^{(2)}}=\frac{|\mathbf{w}_k^T\mathbf{h}_k^{(2)}|}{|\mathbf{w}_k^T\mathbf{h}_k^{(1)}|+|\mathbf{w}_k^T\mathbf{h}_k^{(2)}|}.
\end{equation}

As $\mathcal{H}_k$ is convex, such  $\mathbf{\bar{h}}_k  \in \mathcal{H}_k$. Therefore, due to the noise term, C1 is not fulfilled, and consequently, $\mathbf{w}_k$ cannot be a feasible solution.
\end{proof}

As  $\mathbf{w}_k^T\mathbf{h}_k$ is either always negative or always positive for all $\mathbf{h}_k \in \mathcal{H}_k$ for a feasible value of $\mathbf{w}_k$, we can impose, without loss of generality, that $\mathbf{w}_k^T\mathbf{h}_k >0$. Note that if the optimum precoder, ${\mathbf{w}_k^\star}$, was such that ${\mathbf{w}_k^\star}^T\mathbf{h}_k$ was negative, we could multiply $\mathbf{w}_k^\star$ by $(-1)$ without changing the values of the objective and constraint functions. Imposing strict positivity for any $\mathbf{h}_k \in \mathcal{H}_k$ allows us to remove the absolute value in the right-hand side of C1 in \eqref{eq:alternative_C1}. After removing that absolute value, we have $K$ infinite sets of convex constraints, because a convex constraint has to be fulfilled for a infinite set of channels $\mathbf{h}_k \in \mathcal{H}_k$ with $k=1,...,K$. To deal with this complexity, consider now the following lemma.

\begin{lemma} \label{lemma:2}
Let us take a convex region formulated as $\mathcal{H}_k=\mbox{coh}(\{\mathbf{h}_k^{(j)},j=1,\ldots,J_k\})$ and a function $f_{\mathbf{w}_k}(\mathbf{h}_k)$ parameterized by $\mathbf{w}_k$ and convex w.r.t. $\mathbf{h}_k$. If $f_{\mathbf{w}_k}(\mathbf{h}_k^{(j)})\leq 0$, $j=1,\ldots,J_k$, then $f_{\mathbf{w}_k}(\mathbf{h}_k)\leq 0,\forall \mathbf{h}_k\in\mathcal{H}_k$. 
\end{lemma}

\begin{proof} As the uncertainty region $\mathcal{H}_k$ is the convex hull of their $J_k$ vertices, any element in $\mathcal{H}_k$ can be expressed as a linear convex combination of their vertices:
 \begin{equation}
 \mathbf{h}_k = \sum_{j=1}^{J_k} \alpha_j  \mathbf{h}_k^{(j)}, \quad \alpha_j>0, \quad \sum_{j=1}^{J_k} \alpha_j =1.
 \end{equation}

 Consider now the following function \begin{equation}
  f_{\mathbf{w}_k}(\mathbf{h}_k)=\sqrt{\sigma_k^2+\sum_{i\neq k}\rho^2 (\mathbf{h}_k^T\mathbf{w}_i)^2} -  \frac{\rho}{\sqrt{\gamma_k}}\mathbf{h}_k^T\mathbf{w}_k
  \end{equation}
  that is convex w.r.t. $\mathbf{h}_k$. Due to its convexity, it fulfills:
  \begin{equation}
 f_{\mathbf{w}_k}(\mathbf{h}_k) =  f_{\mathbf{w}_k}\left(\sum_{j=1}^{J_k} \alpha_j  \mathbf{h}_k^{(j)}\right) \leq \sum_{j=1}^{J_k} \alpha_j f_{\mathbf{w}_k}(\mathbf{h}_k^{(j)}) .
  \end{equation}
 Therefore, if $f_{\mathbf{w}_k}(\mathbf{h}_k) \leq 0$ for all the vertices $\{\mathbf{h}_k^{(j)}\}$, then $f_{\mathbf{w}_k}(\mathbf{h}_k)$ will be less than or equal to 0 for all  $\mathbf{h}_k \in \mathcal{H}_k$.
\end{proof}

Thanks to Lemma \ref{lemma:2}, we can re-write problem \eqref{eq:imperfect_CSI_0} as:
 
\begin{equation}
\begin{aligned}
 \underset{\{\mathbf{w}_k\},v}{\text{minimize  }}
 & v \\
 \text{subject to  }
 & C1:\sqrt{\sigma_k^2+\sum_{i\neq k}\rho^2 |\mathbf{w}_i^T\mathbf{h}_k^{(j)}|^2} \leq  \frac{\rho}{\sqrt{\gamma_k}} \mathbf{w}_k^T\mathbf{h}_k^{(j)},    \quad j=1,...,J_k, \quad k = 1, \ldots, K,\\
 & C2: \sum_k{A_k|\mathbf{e}_l^T\mathbf{w}_k|} \leq v, \quad l = 1, \ldots, L,\\& C3: 0\leq v \leq \min(\beta, P_{\max}-\beta), \\
 & C4: \mathbf{w}_k^T\mathbf{h}_k^{(j)} \geq 0, \quad j=1,...,J_k, \quad k = 1, \ldots, K.
\end{aligned}
\label{eq:imperfect_CSI_vertices}
\end{equation}

Note that we have been able to ensure the positiveness of $\mathbf{w}_k^T\mathbf{h}_k$ in the whole uncertainty region by just forcing it at its vertices through the finite set of $\sum_{k=1}^Kk_i$ convex constraints in C4. As a result, the original problem turns out into a \ac{socp} with a finite number of constraints for which several extremely efficient algorithms and tools are available, such as the interior point methods and the SeDuMi software package \cite{Boyd04,sedumi}.

A problem similar to \eqref{eq:imperfect_CSI_0} was considered in Ref.  \cite{bengtsson2002} (Section 18.5.1) for  RF \ac{mumiso} channels, for which channel and beamformers are complex-valued vectors. First, the authors considered perfect \ac{csi} and solved the problem by forcing the imaginary part of $\mathbf{w}_k^T\mathbf{h}_k$ to be zero and the real part to be positive (Equation~(18.29) in \cite{bengtsson2002}). Then, they extended the solution to a robust scheme modeling the uncertainty through the lower and upper bounds of the channel correlation matrix (Equation~(18.4) in Ref.  \cite{bengtsson2002}). Note that this model does not fit the case where the channel uncertainty comes from quantization, which results in uncertainty regions different from the ones considered in Ref. \cite{bengtsson2002}. In addition, the strategy followed in Ref.  \cite{bengtsson2002} for the perfect CSI case cannot be directly applied when having an infinite set of channels (a region) in the constraints instead of a single channel. For RF channels, unless all the channels are 0, it is not possible to ensure a null imaginary part (and a positive real part) of $\mathbf{w}_k^T\mathbf{h}_k$ for all the channels in the region. Fortunately, in \ac{vlc} systems, the channels do not have imaginary parts, and we have proved in Lemma \ref{lemma:1} that the real part of $\mathbf{w}_k^T\mathbf{h}_k$ cannot change its sign within the whole uncertainty region.

\section{Implementation Aspects}
\label{sec:Implementation}
A practical implementation of the proposed scheme requires each user to estimate the \ac{dl} channel from each LED and send a quantized version of this information through a feedback channel. In this section, we discuss some issues regarding the quantization and the impact on the complexity of the optimization problem presented in Sections \ref{sec:Problem} and \ref{sec:Solution}.

The most straightforward quantization strategy is to independently quantize the $L$ channel components (\emph{scalar quantization}) using a uniform quantizer of $B$ bits per channel component. However, depending on the channel amplitudes distribution, some strategies can be taken to reduce the quantization error power without increasing the number of bits. For example, if lower amplitudes are more likely than higher amplitudes, more efficient usage of the available representation levels is achieved if a non-uniform quantizer is employed (for example, a uniform quantizer of the logarithm of the channel). 

For uniform or non-uniform scalar quantization with $B$ bits per channel component, the number of possible uncertainty regions is $N=2^{L\cdot B}$. As any uncertainty region reported has $2^L$ vertices, the total number of constraints represented by C1, $K\cdot 2^L$, increases linearly with $K$ and exponentially with $L$.

On the other hand, in the scenario considered, some combinations of the amplitudes of the \ac{led}s are not possible. For example, as each \ac{led} location is different, a user cannot have a maximum channel simultaneously from all LEDs. A joint quantization of channel components (\emph{vector quantization}) can exploit this feature to reduce the number of CSI information bits. The total number of constraints will depend again on the number of vertices of the uncertainty regions. Nevertheless, the vector quantization case is outside the scope of this paper and we will restrict the simulations to the scalar quantization case.

\section{Simulation Results}
\label{sec:Results}
In this section, we present some simulation results. We consider $L=6$ active LEDs serving $K$ users (see Figure~\ref{fig:escenario} as an example of the simulated setup, similar to the one considered in  Ref. \cite{RoserGlobecom}). We set the target \ac{snir} to $15$ dB for each user, the maximum power threshold to 20 W, and take the rest of the parameters from Ref. \cite{RoserGlobecom}. All the simulation results presented in this section are averaged over different scenario realizations. For each scenario realization, the LEDs are positioned as shown in Figure \ref{fig:escenario} with a fixed height of 2.4 m. The users' (PDs) locations are random. The $x$ and $y$ coordinates of the PDs are drawn from a uniform random distribution in the simulated area, whereas the $z$ coordinate (height of PD) is drawn from a uniform random distribution between 0.5 and 1 m. The channels are then generated according to the relative positions between transmitters (LEDs) and receivers (PDs).

\begin{figure}[t]
	\centering
		 \includegraphics[trim = 0mm 0mm 0mm 0mm, clip=true, width= 0.8\columnwidth]{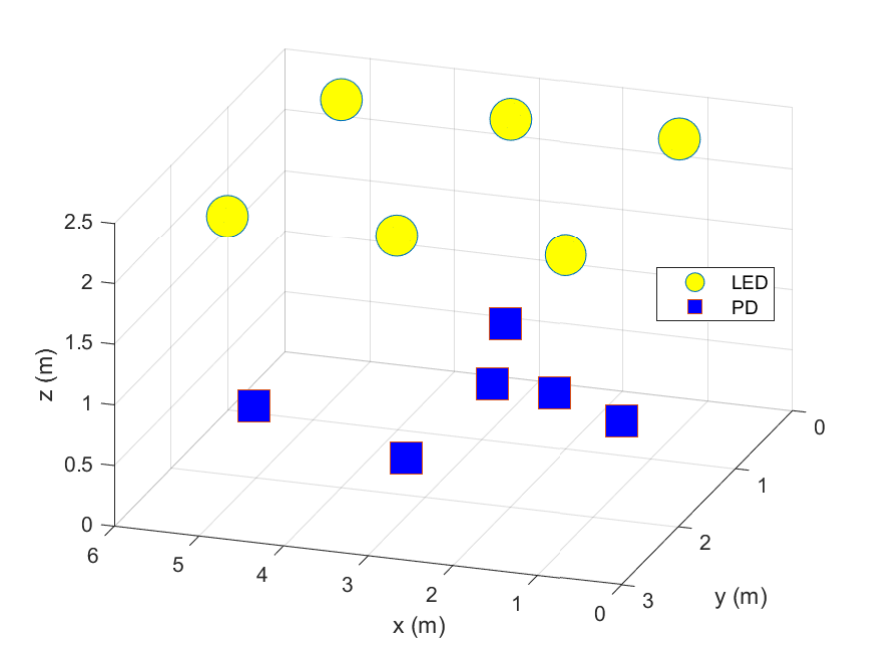}

	\caption{Distribution of light emitting diodes (LEDs) (fixed positions) and photo diodes (PDs) (random positions) in a 3D space.}
	\label{fig:escenario}
\end{figure}

For each scenario realization, we design precoders that serve $K$ users jointly. The quantized channel is the only CSI available at the transmitter side to design robust and non-robust precoders. For the simulations, we have considered the particular case of independent quantization of the channel between each LED and each PD, that is, scalar quantization. Each of these channel coefficients is a real positive magnitude that can take a wide margin of values. To cope with this wide margin, we perform a logarithm transformation and then a uniform quantization. In other words, we consider a uniform quantization of the value in dB of the channel magnitude. The dynamic margin of the quantizer, i.e., the minimum and maximum values represented by the quantizer, are determined statistically using $10^6$ realizations of users' positions and channel values of the setup considered in the simulations. We design the non-robust precoder using the quantized channel as if it were the actual channel (despite being different). The actual channel is within an uncertainty region around the quantized channel. As explained in the previous sections, we design the robust precoder to achieve the target \ac{snir} for any channel within this region. For the computation of the non-robust and robust solutions, we have used CVX, a package for specifying and solving convex programs \cite{cvx, gb08}.

The number of bits impacts the size of the channel uncertainty regions, making the constraints of the robust precoder design problem more challenging. As a result, a higher optical power is needed to fulfill the constraints. Note that, depending on the channel realizations, fulfilling the constraints may require an unacceptable amount of power or not be possible at all; that is, it is not possible to fulfill constraint C3. In Figure~\ref{fig:exito}, we show the percentage of feasibility (i.e., successful designs) for different numbers of jointly served users and quantization bits. As expected, the percentage reduces (i.e., the precoder design becomes more complex) as the number of bits decreases or as the number of users increases. 

\begin{figure}[t]
	\centering
	 \includegraphics[trim = 0mm 0mm 0mm 0mm, clip=true, width= 0.8\columnwidth]{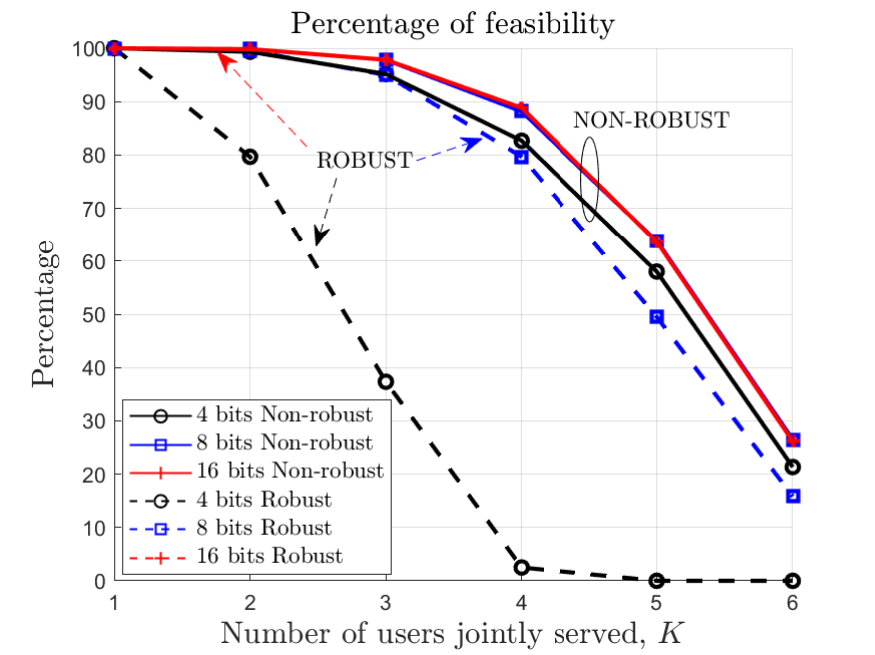}
	
	\caption{Percentage of successful designs (averaged over different channel realizations) for different numbers of jointly served
users and quantization bits.
}
	\label{fig:exito}
\end{figure}

It is not surprising that the percentage of feasibility cases of the robust design is lower than for the non-robust design due to the much harder constraints of the robust approach. Note, however, that the non-robust solution will perform much worse than the robust counterpart for channels different from the quantized version. To illustrate this, Figure~\ref{fig:worst_corner} shows the \ac{snir} for the worst possible channel, which always corresponds to a corner of the uncertainty region ($\mathbf{h}_k^{(j)}, j=1,...,J_k$, $k = 1, \ldots, K$ of constraint C1), as proved in this paper, for the non-robust and robust approaches. Therefore, if a feasible solution exists, the robust precoder fulfills the \ac{snir} constraint with equality at this corner by design, which is set to 15 dB in these simulations, and, consequently, in all of the uncertainty region. As a result, as shown in Figure~\ref{fig:worst_corner}, for the robust precoder, the \ac{snir} at the worst-possible channel of the uncertainty region is always equal to the target \ac{snir}. In contrast, the non-robust design cannot guarantee that, especially when the number of quantization bits is low.

\begin{figure}[t]
	\centering
	 \includegraphics[trim =0mm 0mm 0mm 0mm, clip=true, width= 0.8\columnwidth]{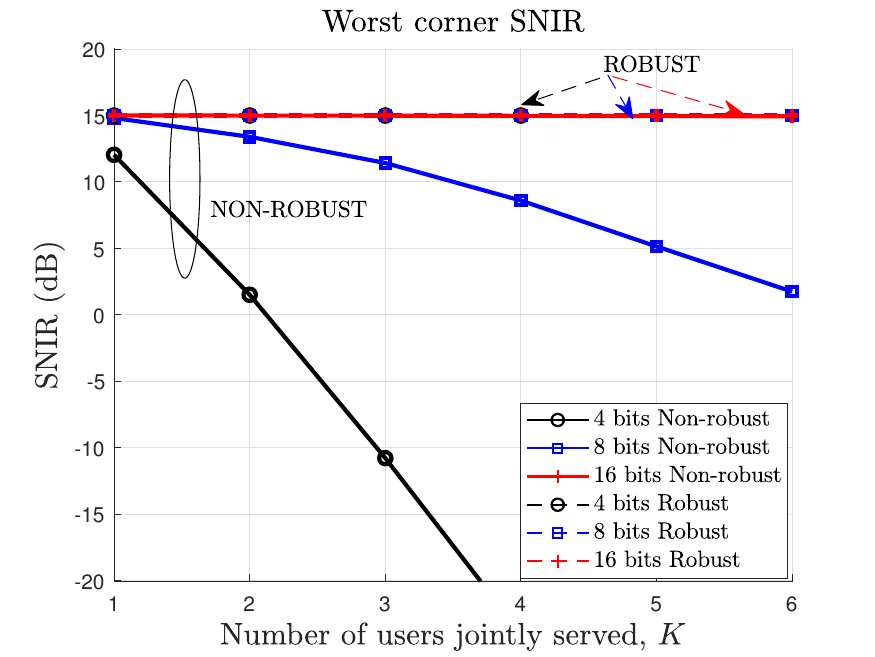}
	
	\caption{Signal to noise plus interference ratio (SNIR) at the worst corner of the uncertainty regions of the $K$ users group (averaged over different channel realizations). The target \ac{snir} is 15 dB.  
	}
	\label{fig:worst_corner}
\end{figure}
	
In practice, we may be interested in finding how well the feasible designs will perform for the real channels experienced by the users. To that end, we compute the \ac{snir} using the actual channels and the precoders designed using the quantized channels (as the actual channels information is not available at the design phase). In Figure~\ref{fig:snr_real}, we show the performance of the robust and non-robust precoders in terms of the actual \ac{snir} of the worst user within each group of $K$ users, averaged over different channel realizations. The actual \ac{snir} is the \ac{snir} computed with the actual channels, which are different to the quantized channels used for the beamformer design. For the non-robust precoder, the achieved \ac{snir} is below the target \ac{snir}. This is because the non-robust design incorrectly assumes that the quantized channel is equal to the actual channel. On the other hand, the robust precoder is designed to achieve the target \ac{snir} for any channel within the uncertainty region. The actual channel always lies within the uncertainty region. However, it is not necessarily the worst-case channel, which is located at a corner of the uncertainty region, as shown in previous sections. Consequently, with the robust beamformer, the actual users' SNIR can be above the target \ac{snir}, particularly if the uncertainty region increases, which happens when the number of quantization bits decreases. As illustrated previously in Figure~\ref{fig:exito}, for 4 bits and 5 users or more, no robust precoder can achieve the target \ac{snir} with a power lower than or equal to $P_{\max}$. This is the reason why no \ac{snir} values are shown for these cases in Figures~\ref{fig:worst_corner} and \ref{fig:snr_real}. Note that this does not invalidate our approach. It just tells us that for 5 users or more, we need to increase the number of bits to fulfill the target \ac{snir} under channel uncertainty with a power lower than or equal to $P_{\max}$.

\begin{figure}[t]
	\centering
		 \includegraphics[trim = 0mm 0mm 0mm 0mm, clip=true, width= 0.8\columnwidth]{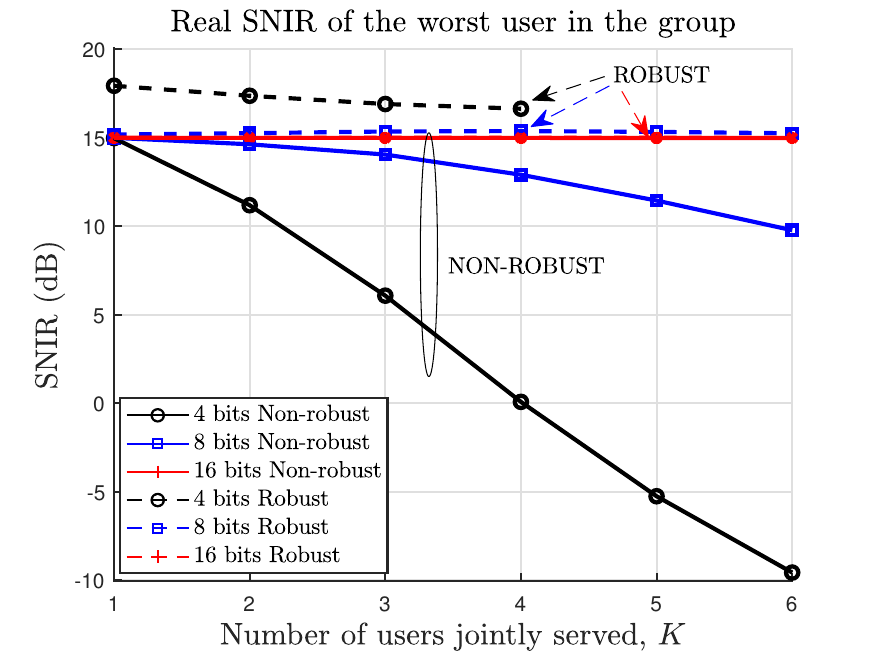}

	\caption{\ac{snir} of the worst user within each group of $K$ users (averaged over different channel realizations). The target \ac{snir} is 15 dB. 
	}
	\label{fig:snr_real}
\end{figure}

To ensure feasibility, given the number of bits $B$, we can reduce the number of simultaneously served users $K$ until a feasible precoder can be found. Moreover, as proved in Section \ref{sec:Solution}, we can guarantee that the robust precoder will fulfill the target quality for all the users in the set. Table \ref{tab:SNR} shows the \ac{snir} of the worst user within each set of jointly served users (with the maximum number of users possible provided that a feasible solution can be found). Note that, in the case of the non-robust precoder, the target quality cannot be achieved. More importantly, it is impossible to decide beforehand if the non-robust precoders will or will not do the job because only the quantized channel information (and not the actual channel) is available at the design stage.

\begin{table}
\caption{\label{tab:SNR}Real SNIR (dB) of the worst user within each feasible group (averaged over different channel realizations).}
\begin{center}
\begin{tabular}[t]{|c|c|c|c|}
\hline
\textbf{Number of Bits} & \boldmath{$B=4$} & \boldmath{$B=8$}  & \boldmath{$B=16$} \\
Non-robust design & $-3.5459$  & $11.3752$  & $14.9915$\\
Robust design & $17.3103$  & $15.3762$  & $15.0025$  \\ 
\hline
\end{tabular}
\end{center}
\end{table}

\section{Conclusions and Discussion}
\label{sec:Conclusions}
In this paper, we have designed the optimum \ac{mumiso} precoder robust against channel quantization errors in \ac{vlc} systems that minimizes the transmitted optical power while guaranteeing a target \ac{snir} per user. We show that the exact form for this precoder can be obtained as the solution of an \ac{socp} problem. As expected, the precoder design becomes more difficult (i.e., requires more power) as the number of bits decreases or as the number of users increases. Nevertheless,  provided the budget power is sufficient, we can guarantee that the robust designs achieve the target \ac{snir} per user. In contrast, for non-robust designs, such a guarantee is not possible.

As future work, we will consider the robust precoding design problem in the presence of both quantization noise and channel estimation errors. This could include also the case of vector quantization as a way to improve the quality of the \ac{csi} without increasing the number of quantization bits. Another research line will be the extension of the proposed scheme to the \ac{mimo} case, where each receiver has more than one \ac{pd}. Finally, and given the non-linear responses of the \ac{led}s and the \ac{pd}s, it will be convenient to study how these non-linearities affect the proposed designs and analyze how to modify the designs to be robust also against them.

\section{Materials and Methods}
\label{sec:materials}
The results in this paper can be replicated by implementing the algorithms described in it. The authors disclose that the code is protected and cannot be made public. Interested readers can get in touch with the authors to receive assistance.

\section{Author Contributions and Conflicts of Interest}
Conceptualization, O.M. and A.P.-I.; methodology, O.M. and A.P.-I.; software, O.M. and G.S.A.; validation, O.M. A.P.-I., and G.S.A.; formal analysis, O.M. and A.P.-I.; investigation, O.M. and A.P.-I.; writing---original draft preparation, O.M., A.P.-I., and G.S.A.; visualization, O.M.; supervision, O.M. and A.P.-I.; project administration, A.P.-I.; funding acquisition, A.P.-I. All authors have read and agreed to the published version of the manuscript.

The authors declare no conflict of interest. The funders had no role in the design of the study; in the collection, analyses, or interpretation of data; in the writing of the manuscript; or in the decision to publish the results.


\section{Abbreviations}
The following abbreviations are used in this manuscript:\\

\noindent 
\begin{tabular}{@{}ll}
CSI & Channel state information\\
DC & Direct current\\
DL & Downlink\\
FDD & Frequency division duplexing\\
LED & Light emitting diode\\
MIMO & Multiple-input multiple-outpout\\
MSE & Mean square error\\
MU-MISO & Multi-user multiple-input single-output\\
PD & Photo diode\\
RF & Radio frequency\\
SER & Symbol error rate\\
SNIR & Signal to noise plus interference ratio\\
SNR & Signal to noise ratio\\
SOCP & Second order cone programming\\
TDD & Time division duplexing\\
UL & Uplink\\
VLC & Visible light communication\\
w.r.t. & with respect to\\
ZF & Zero forcing
\end{tabular}

\end{document}